\documentclass[a4paper,11pt]{article}

\usepackage{amsfonts}
\usepackage{amsmath}
\usepackage{amssymb}
\usepackage{amsthm}
\usepackage{url}
\newtheorem{theorem}{Theorem}
\newtheorem{proposition}[theorem]{Proposition}

\newtheorem{lemma}[theorem]{Lemma}
\newtheorem{corollary}[theorem]{Corollary}
\newtheorem{remark}{Remark}
\newtheorem{conjecture}{Conjecture}
\newtheorem{example}{Example}

\newcommand{\F}{\mathbb{F}}

\newcommand{\be}{\begin{eqnarray}}
\newcommand{\ee}{\end{eqnarray}}
\newcommand{\nn}{{\nonumber}}
\newcommand{\dd}{\displaystyle}

\newcommand{\hull}{\mbox{\rm Hull}}

\begin{document}
\date{}
\title{Sequence of Numbers of Linear Codes with Increasing Hull Dimensions}
\author{Stefka Bouyuklieva \thanks{Stefka Bouyuklieva is with the Faculty of Mathematics and Informatics, St. Cyril and St. Methodius University of Veliko Tarnovo, 5000
Veliko Tarnovo, Bulgaria, (email: stefka@ts.uni-vt.bg).},~
Iliya Bouyukliev \thanks{Iliya Bouyukliev is with the Institute of Mathematics and Informatics, Bulgarian Academy of Sciences, 5000 Veliko Tarnovo, Bulgaria, (email: iliyab@math.bas.bg).},~
and Ferruh \"{O}zbudak \thanks{Ferruh \"{O}zbudak is with the Faculty of Engineering and Natural Sciences, Sabanc{\i} University, 34956 Istanbul, T\"{u}rkiye, (email: ferruh.ozbudak@sabanciuniv.edu).} }
\maketitle

\begin{abstract}
The hull of a linear code $\mathcal{C}$ is the intersection of $\mathcal{C}$ with its dual code. We present and analyze the sequence of numbers of linear codes with increasing hull dimension but
the same length $n$ and dimension $k$. We also present classification results for binary and ternary linear codes with trivial hulls (LCD and self-orthogonal) for some values of the length $n$ and dimension $k$, comparing the obtained numbers with the number of all linear codes for the same $n$ and $k$.
\end{abstract}

\section{Introduction}
\label{intro}

Let $\F_q$ be a finite field with $q$ elements, and $n$ be a positive integer.
We consider $\F_q^n$ as an $n$-dimensional vector (linear) space over $\F_q$.
Let $\cdot$ be the Euclidean inner product on $\F_q^n$ given by
\be
(a_1, \ldots, a_n) \cdot (b_1,\ldots, b_n) =a_1b_1 + \cdots + a_nb_n.
\nn\ee
Any linear subspace $\mathcal{C}$ of $\F_q^n$ is called a {\it linear $q$-ary code} of {\it length} $n$ and {\it dimension} $k$, where $k=\dim \mathcal{C}$.
We say that $\mathcal{C}$ is an $[n,k]_q$ code in short. Throughout the paper we consider only $\F_q$-linear codes in $\F_q^n$.

For an $[n,k]_q$ code $\mathcal{C}$, let $\mathcal{C}^\perp$ be the subset
\be
\mathcal{C}^\perp=\{y \in \F_q^n: x \cdot y=0 \; \mbox{for all} \; x \in \mathcal{C}\},
\nn\ee
where we use the Euclidean inner product. As it is a non-degenerate inner product, we conclude that $\mathcal{C}^\perp$ is an $[n,n-k]_q$ code. Moreover $\mathcal{C}^\perp$ is called the {\it dual code} of $\mathcal{C}$.

The {\it hull} of a linear code $\mathcal{C}$ is the linear space $\mathcal{C} \cap \mathcal{C}^\perp$, that we denote as $\hull(\mathcal{C})$. Note that $\hull(\mathcal{C})=\hull(\mathcal{C}^\perp)$. The hull of a linear code $\mathcal{C}$ was introduced by Assmus, Jr. and Key \cite{Assmus-Key}, where they used it for classification of finite projective planes. The study of hulls of linear codes has many interesting applications ranging from the construction of quantum error correcting codes to post-quantum cryptography. We refer, for example, to \cite{quantum+hull} and the references therein for further details of such connections and applications.

Two extreme cases of hulls of linear codes have been studied separately. If $\dim \hull(\mathcal{C})=0$, then $\mathcal{C}$ is called a {\it linear complementary dual (LCD)} code. LCD codes were introduced  by Massey \cite{Massey} in 1992 in the framework of coding theory. Sendrier proved that LCD codes meet the Gilbert-Varshamov bound \cite{Sendrier GV}. In \cite{Carlet_Guilley}, Carlet and Guilley investigated an  application of LCD codes against side-channel attacks (SCA) and fault injection attacks (FIA) and obtained further interesting results. This motivated further research and the study of LCD codes have been a very active area in the recent years. We refer, for example, to \cite{Carlet_Pellikaan,Dougherty-Kim-Sole-LCD,Kim-LCD,Harada-LCD} and the related references.

Another extreme case is when  $\dim \hull(\mathcal{C})= \dim \mathcal{C}$, which means that $\mathcal{C} \subseteq \mathcal{C}^\perp$. In such a case $\mathcal{C}$ is called {\it self-orthogonal (SO)} code.
A very important particular case of SO codes is the case $\mathcal{C}=\mathcal{C}^\perp$, which holds only if $n$ is even and $n=2k$, where $\mathcal{C}$ is an $[n,k]_q$ code. Such $q$-ary linear codes are called {\it self-dual} codes. The study of self-dual codes has very interesting connections and applications to various breaches of mathematics.  We refer, for example, to the influential books \cite{Conway_Sloane, Nebe} for further details on this topic.

As $\hull(\mathcal{C})=\hull(\mathcal{C}^\perp)$, we assume that $\mathcal{C}$ is an $[n,k]_q$ code with $k \le n/2$ throughout this paper. For integers $0 \le \ell \le k$, let $A_{n,k,\ell,q}$ denote the number of $q$-ary codes $\mathcal{C} \subseteq \F_q^n$ such that $\dim \mathcal{C}= k$ and $\dim \hull(\mathcal{C})=\ell$. In \cite{Sendrier_hull}\ Sendrier established an important mass formula giving $A_{n,k,\ell,q}$, that we also refer to Theorem \ref{theorem.Sendrier} below for its statement. The formula of \cite{Sendrier_hull} is quite complicated including the number of self-orthogonal codes of various dimensions and complex sums.

In the first part of this paper, we focus on the following sequence of numbers of linear codes with increasing hull dimensions,
which comprehensively encapsulates the number of
$q$-ary $[n,k]$ codes
(and equivalently  $[n,n-k]$ codes) for all possible hull dimensions in a structured manner.
\be \label{sequence S}
\mathcal{S}_{n,k,q}:=\left(A_{n,k,0,q},A_{n,k,1,q}, \cdots, A_{n,k,\ell,q}, \ldots, A_{n,k,k,q}\right).
\ee
Note that $\mathcal{S}_{n,k,q}$ is a finite sequence of integers, and its length is $k+1$. Using numerical results we observe that the sequence in (\ref{sequence S}) is always {\bf strictly decreasing} independent from $n$, $k$ or $q$. In fact, in an earlier version of this paper, using \cite{Sendrier_hull} and a detailed study of the corresponding complicated mathematical formulae, it was proved that
\be \label{result decreasing}
A_{n,k,0,q} > A_{n,k,1,q} > \cdots > A_{n,k,\ell,q} > \cdots > A_{n,k,k,q}.
\ee

One of the main results of this paper is a  far reaching improvement of the result in (\ref{result decreasing}). For integers $0 \le \ell \le k-1$, let $\mu_{n,k,\ell,q}$ be the largest integer such that
\be 
A_{n,k,\ell,q} \ge \mu_{n,k,\ell,q} A_{n,k,\ell+1,q}.
\nn\ee
Note that  the result in (\ref{result decreasing}) is equivalent to the statement that for integers $0 \le \ell \le k-1$ we have
\be \label{result2 decreasing}
\mu_{n,k,\ell,q} \ge 1
\ee
Namely in Theorem \ref{thm:ineq} below we prove that for integers $0 \le \ell \le k-1$ we have
\be \label{result_new}
\mu_{n,k,\ell,q} \ge \left\{
\begin{array}{rl}
q^{\ell+1} -1 & \mbox{if Condition * does not hold}, \\ \\
\dd \frac{q^{\ell+1}-1}{2} & \mbox{if Condition * holds}
\end{array}
\right.
\ee
Here Condition * is the condition that
\be
\mbox{$q$ is odd, $n \equiv 2 \pmod 4$, and $-1$ is not a square in $\F_q^*$}.
\nn\ee
 Moreover, we determine when the bound in (\ref{result_new}) is tight.

 In our proofs we use a recent simple and alternative mass formulae for the integers $A_{n,k,\ell.q}$ presented in \cite{LiShiLing}.
 We refer to Theorems \ref{theorem.Li et al.even} and \ref{theorem.Li et al.odd} below for the statements of the corresponding formulae in certain cases. We observe that the formulae in \cite{Sendrier_hull} and \cite{LiShiLing} are very different in their presentations. For example, the formula in \cite{Sendrier_hull} consists of certain inner summations, and hence it is a sum formula. However the formulae in \cite{LiShiLing} consist of certain inner multiplies, and hence they are multiplication formulae.
 In our problem, it turns out that the result in \cite{LiShiLing}
is much more suitable for studying the sequence $\mathcal{S}_{n,k,q}$ providing improvement in (\ref{result_new}).
The result in (\ref{result_new}) requires a long and detailed study we present in Section \ref{sec:proofs} below.
We could not establish such a strong result using the mass formula of \cite{Sendrier_hull}.

In the second part of our paper we classify all linear, self-orthogonal and LCD codes over $\F_2$ and $\F_3$ for different lengths $n\le 20$ and dimensions $k\le 10$. If for given $n$ and $k$ the number of all inequivalent linear $[n,k]$ codes is huge, we give a classification result only for the optimal codes with the considered length and dimension.

This paper is organized as follows. In Section \ref{sect:preliminaries}, we present the needed definitions, theorems and formulae to prove our result. Section \ref{sec:proofs} is devoted to the proof of the main theorem, namely Theorem \ref{thm:ineq}. Note that $A_{n,k,\ell,q}$ is the number of all different codes with the corresponding parameters, but the number of the inequivalent among them is much smaller. Since we do not have a formula for this number, we did some experiments and classified the binary and ternary linear $[n,k]$ codes for some values of $n$ and $k$. We present and analyze the obtained results in Section \ref{Results}. We end the paper with three tables consisting of enumeration results for inequivalent linear codes with different lengths and dimensions over the fields with two and three elements.


\section{Preliminaries}
\label{sect:preliminaries}

 In this section, we present the basic definitions, theorems and formulae that we need for the proof of the main theorem, as well as for the computational results in Section \ref{Results}.

As we already defined in the Introduction, a linear $[n,k,d]$ $q$-ary code $\mathcal{C}$ is a $k$-dimensional subspace of the
vector space $\F_q^n$, and $d$ is the
smallest weight among all non-zero codewords of $\mathcal{C}$, called the \textit{minimum weight} (or minimum distance) of the code. A code is called {\em even} if all its codewords have even weights. The orthogonal complement of $\mathcal{C}$ according to the defined (in our case Euclidean) inner product is called the \textit{dual code} of $\mathcal{C}$ and denoted by $\mathcal{C}^{\perp}$. The intersection $\mathcal{C}\cap \mathcal{C}^\perp$ is called the hull of the code. As mentioned in the introduction, the dimension of the hull can be at least 0 and at most $\min\{k,n-k\}$.


Two linear $q$-ary codes $\mathcal{C}_1$ and $\mathcal{C}_2$ are equivalent if the codewords of $\mathcal{C}_2$ can be obtained from
the codewords of $\mathcal{C}_1$ via a sequence of transformations of the
following types:
\begin{itemize}
\item[(1)] permutation on the set of coordinate positions;
\item[(2)] multiplication of the elements in a given position by a
non-zero element of $\F_q$;
\item[(3)] application of a field automorphism to the elements in all
coordinate positions.
\end{itemize}
An automorphism of a linear code $\mathcal{C}$ is a sequence of the
transformations $(1)-(3)$ which maps
each codeword of $\mathcal{C}$ onto a codeword of the same code. The set of all automorphisms of $\mathcal{C}$ forms a group called the automorphism group $Aut(\mathcal{C})$ of the code. The presented formulae do not count the equivalence but in Section \ref{Results} we enumerate only inequivalent codes.

If the considered inner product is Euclidean, the dimension of the hull  is an invariant under the defined equivalence relation in the cases $q = 2$ and $q = 3$, but if $q\ge 4$, any linear code over
$\F_q$ is equivalent to an Euclidean LCD code \cite{Carlet_Pellikaan}. If $q=p^{2s}$, where $p$ is the characteristic of the field, we can consider the Hermitian inner product over $\F_q$ defined by $$(x,y)=\sum_{i=1}^n x_iy_i^{\sqrt{q}}, \ \forall x=(x_1,\ldots,x_n), y=(y_1,\ldots,y_n)\in \F_q^n.$$
For this inner product, the dimension of the hull is an invariant in the case of the quaternary codes ($q=4$), and if $q>4$, any linear code over
$\F_q$, $q=p^{2s}$, is equivalent to an Hermitian LCD code \cite{Carlet_Pellikaan}. In this paper, we consider only the Euclidean inner product.

Formulae for counting the number $\sigma(n,k)$ of all different $q$-ary self-orthogonal codes of length $n$ and dimension $k$
were proven in \cite{Pless_sigma}. 


\begin{theorem}\label{thm:sigma}
Let $m=\lfloor n/2\rfloor$ and $\pi_{n,k}=\prod_{i=1}^k\frac{q^{2m-2i+2}-1}{q^i-1}$. For all $1\le k\le m$, we have
\[ \sigma_{n,k}=\left\{ \begin{array}{rl}
\pi_{n,k},&\mbox{if} \ n \ \text{is odd},\\
\frac{q^{n-k}-1}{q^n-1}\pi_{n,k},&\mbox{if} \ n \ \mbox{and} \ q \ \text{are even},\\
\frac{q^{m-k}-1}{q^m-1}\pi_{n,k},&\mbox{if} \ n\equiv 2\pmod 4 \ \mbox{and} \ q\equiv 3\pmod 4,\\
\frac{q^{m-k}+1}{q^m+1}\pi_{n,k},&\mbox{if} \ (n\equiv 0\pmod 4)\\
& \ \mbox{or} \ (n\equiv 2\pmod 4 \ \mbox{and} \ q\equiv 1\pmod 4).
\end{array}\right.\]
\end{theorem}

Note that $\sigma_{n,0}=1$. 
%
%
Formulae for the number of all linear $[n,k]_q$ codes with hull dimension $\ell$ are proved by Sendrier.

\begin{theorem}\cite{Sendrier_hull} \label{theorem.Sendrier}
Let $k$ and $n\ge 2k$ be positive integers. If $\sigma_{n,i}$ is the number of all self-orthogonal $[n,i]_q$ codes, $i\le k$, then the number of all $[n,k]_q$ codes whose hull has dimension $\ell$, $\ell\le k$, is equal to
\begin{equation}
A_{n,k,\ell,q}=\sum_{i=l}^k {n-2i\brack k-i}_q{i\brack l}_q(-1)^{i-l}q^{i-l\choose 2}\sigma_{n,i}.
\end{equation}
\end{theorem}

In the above formula, ${n\brack k}_q$ is the Gaussian binomial coefficient, defined by
\begin{equation}
{n\brack 0}_q=1, \ {n\brack k}_q=\frac{(q^n-1)(q^{n-1}-1)\ldots(q^{n-k+1}-1)}{(q^k-1)(q^{k-1}-1)\ldots(q-1)}, \mbox{if} \ 1\le k\le n.
\end{equation}

We use the following properties of the Gaussian coefficients for $0\le k\le n-1$:
\begin{eqnarray}
{n\brack n-k}_q &=& {n\brack k}_q, \\
  {n+1\brack k}_q &=& \frac{q^{n+1}-1}{q^{n-k+1}-1}{n\brack k}_q,\\
  {n\brack k+1}_q &=& \frac{q^{n-k}-1}{q^{k+1}-1}{n\brack k}_q, \\
  {n+1\brack k+1}_q&=&\frac{q^{n+1}-1}{q^{k+1}-1}{n\brack k}_q.
\end{eqnarray}

In \cite{LiShiLing}, the authors simplified the formula for $A_{n,k,\ell,q}$ separately for even and odd $q$.

\begin{theorem}\cite[Theorem 3.4]{LiShiLing} \label{theorem.Li et al.even} Assume that $q$ is even. Let $\ell$, $k$ and $n$ be three positive integers such that $\ell\le k\le n-\ell$. Suppose that $k_0=k-\ell$. Then
\begin{equation}\label{evenq}
A_{n,k,\ell,q}=\left\{\begin{array}{ll}
\left(\prod_{i=1}^\ell\frac{q^{n-k_0-i}-q^{i-1}}{(q^i-1)q^{k_0+i-1}}\right)q^{(nk_0-k_0^2+n-1)/2}{n/2-1\brack(k_0-1)/2 }_{q^2}&\mbox{if} \ n \ \mbox{is even}, k_0 \ \mbox{is odd},\\
\left(\prod_{i=1}^{\ell-1}\frac{q^{n-k_0-i}-q^{i}}{(q^i-1)q^{k_0+i-1}}\right)q^{(n-k_0)(k_0-1)/2+n-k}\frac{q^{n-k_0}-1}{{(q^\ell}-1)q^{k-1}}{(n-1)/2\brack(k_0-1)/2 }_{q^2}&\mbox{if} \ n \ \mbox{is odd}, k_0 \ \mbox{is odd},\\
\left(\prod_{i=1}^\ell\frac{q^{n-k_0-i}-q^{i-1}}{(q^i-1)q^{k_0+i-1}}\right)q^{k_0(n-k_0+1)/2}{(n-1)/2\brack k_0/2 }_{q^2}&\mbox{if} \ n \ \mbox{is odd}, k_0 \ \mbox{is even},\\
\left(\prod_{i=1}^{\ell-1}\frac{q^{n-k_0-i}-q^{i}}{(q^i-1)q^{k_0+i-1}}\right)q^{k_0(n-k_0)/2}\frac{q^{n-\ell}-1}{{(q^\ell}-1)q^{k-1}}{n/2-1\brack k_0/2 }_{q^2}&\mbox{if} \ n \ \mbox{is even}, k_0 \ \mbox{is even},\\
\end{array}\right.
\end{equation}
\end{theorem}

As in the above formulae the product goes from 1 to $\ell$ or $\ell-1$, the formulae for $\ell=0$ are different and they follow from Theorem 3.2 in \cite{LiShiLing}.

\begin{theorem}
Assume that $q$ is even. Let $k$ and $n$ be positive integers such that $k< n$. Then
\begin{equation}\label{evenq0}
A_{n,k,0,q}=\left\{\begin{array}{ll}
q^{(nk-k^2+n-1)/2}{n/2-1\brack(k-1)/2 }_{q^2}&\mbox{if} \ n \ \mbox{is even}, k \ \mbox{is odd},\\
q^{(n-k)(k-1)/2+n-k}{(n-1)/2\brack(k-1)/2 }_{q^2}&\mbox{if} \ n \ \mbox{is odd}, k \ \mbox{is odd},\\
q^{k(n-k+1)/2}{(n-1)/2\brack k/2 }_{q^2}&\mbox{if} \ n \ \mbox{is odd}, k \ \mbox{is even},\\
q^{k(n-k)/2}\frac{q^{n}-1}{q^{n-k}-1}{n/2-1\brack k/2 }_{q^2}&\mbox{if} \ n \ \mbox{is even}, k \ \mbox{is even},\\
\end{array}\right.
\end{equation}
\end{theorem}

The formulae \eqref{evenq} and \eqref{evenq0} can be combined and simplified further. In the formulae \eqref{evenq} the product goes from 1 to $\ell$ or $\ell-1$. We have transformed these formulae a little and combined with \eqref{evenq0} so that the product in all of them goes from 0 to $\ell$. Note that in the case $\ell=0$ we assume that the product is equal to 1.

\begin{corollary}\label{cor-evenq}
Assume that $q$ is even. Let $\ell$, $k$ and $n$ be integers, such that $0\le\ell\le k< n-\ell$. Suppose that $k_0=k-\ell$. Then
\[
A_{n,k,\ell,q}=\left\{\begin{array}{ll}
\left(\prod_{i=1}^\ell\frac{q^{n-k_0-i}-q^{i-1}}{(q^i-1)q^{k_0+i-1}}\right)q^{(nk_0-k_0^2+n-1)/2}{n/2-1\brack(k_0-1)/2 }_{q^2}&\mbox{if} \ n \ \mbox{is even}, k_0 \ \mbox{is odd},\\
\left(\prod_{i=1}^{\ell}\frac{q^{n-k_0-i}-q^{i}}{(q^i-1)q^{k_0+i-1}}\right)\frac{q^{n-k_0}-1}{q^\ell(q^{n-k-\ell}-1)}q^{(n-k_0)(k_0-1)/2+n-k}{(n-1)/2\brack(k_0-1)/2 }_{q^2}&\mbox{if} \ n \ \mbox{is odd}, k_0 \ \mbox{is odd},\\
\left(\prod_{i=1}^\ell\frac{q^{n-k_0-i}-q^{i-1}}{(q^i-1)q^{k_0+i-1}}\right)q^{k_0(n-k_0+1)/2}{(n-1)/2\brack k_0/2 }_{q^2}&\mbox{if} \ n \ \mbox{is odd}, k_0 \ \mbox{is even},\\
\left(\prod_{i=1}^{\ell}\frac{q^{n-k_0-i}-q^{i}}{(q^i-1)q^{k_0+i-1}}\right)\frac{q^{n-\ell}-1}{q^{\ell}(q^{n-k-\ell}-1)}q^{k_0(n-k_0)/2}{n/2-1\brack k_0/2 }_{q^2}&\mbox{if} \ n \ \mbox{is even}, k_0 \ \mbox{is even}.
\end{array}\right.
\]
\end{corollary}

\begin{remark} If $k=n-\ell$ which means that $\ell=n-k$ and the code $\mathcal{C}^\perp$ is self-orthogonal, the formulae in Corollary \ref{cor-evenq} are not valid because they contain 0 as a factor in the denominator.
\end{remark}

The formulae for odd characteristic depend on the Legendre character $\eta$ of $\F_q$. Recall that $\eta(c)=1$ if $c$ is a square in $\F_q^*$, and $\eta(c)=-1$ otherwise ($c\neq 0$).

\begin{theorem}\cite[Theorem 4.12]{LiShiLing} \label{theorem.Li et al.odd} Assume that $q$ is odd. Let $\ell$, $k$ ana $n$ be three positive integers such that $\ell\le k\le n-\ell$. Suppose that $k_0=k-\ell$. Then
\begin{equation}\label{oddq}
A_{n,k,\ell,q}=\left\{\begin{array}{ll}
\left(\prod_{i=1}^\ell\frac{q^{n-k_0-2i+1}-1}{q^{k_0}(q^i-1)}\right)q^{(k_0(n-k_0)-1)/2}B_1{n/2-1\brack(k_0-1)/2 }_{q^2}&\mbox{if} \ n \ \mbox{is even}, k_0 \ \mbox{is odd},\\
\left(\prod_{i=1}^{\ell}\frac{q^{n-k_0-2i+2}-1}{q^{k_0}(q^i-1)}\right)q^{(n-k_0)(k_0+1)/2-\ell}{(n-1)/2\brack(k_0-1)/2 }_{q^2}&\mbox{if} \ n \ \mbox{is odd}, k_0 \ \mbox{is odd},\\
\left(\prod_{i=1}^\ell\frac{q^{n-k_0-2i+1}-1}{q^{k_0}(q^i-1)}\right)q^{k_0(n-k_0+1)/2}{(n-1)/2\brack k_0/2 }_{q^2}&\mbox{if} \ n \ \mbox{is odd}, k_0 \ \mbox{is even},\\
\left(\prod_{i=1}^{\ell}\frac{q^{n-k_0-2i+2}-1}{q^{k_0}(q^i-1)}\right)q^{k_0(n-k_0)/2}B_2{n/2\brack k_0/2 }_{q^2}&\mbox{if} \ n \ \mbox{is even}, k_0 \ \mbox{is even},\\
\end{array}\right.
\end{equation}
where $B_1=q^{n/2}-\eta((-1)^{n/2})$ and $B_2=\frac{q^{n/2-\ell}+\eta((-1)^{n/2})}{q^{n/2}+\eta((-1)^{n/2})}$.
\end{theorem}

The following formulae apply to LCD codes:

\begin{theorem} \cite[Corollary 32]{CarletLCD2019} Let $q$ be a power of an odd prime and $k$, $n$
be two positive integers with $k < n$. Then
\[
A_{n,k,0,q}=\left\{\begin{array}{ll}
q^{(k(n-k)-1)/2}(q^{n/2}-\eta((-1)^{n/2})){n/2-1\brack(k-1)/2 }_{q^2}&\mbox{if} \ n \ \mbox{is even}, k \ \mbox{is odd},\\
q^{(k+1)(n-k)/2}{(n-1)/2\brack(k-1)/2 }_{q^2}&\mbox{if} \ n \ \mbox{is odd}, k \ \mbox{is odd},\\
q^{k(n-k+1)/2}{(n-1)/2\brack k/2 }_{q^2}&\mbox{if} \ n \ \mbox{is odd}, k_0 \ \mbox{is even},\\
q^{k(n-k)/2}{n/2\brack k/2 }_{q^2}&\mbox{if} \ n \ \mbox{is even}, k_0 \ \mbox{is even},\\
\end{array}\right.
\]
\end{theorem}

\begin{remark} It is easy to see that the formulae \eqref{oddq} also apply to $\ell=0$.
\end{remark}

\section{Comparing the numbers of linear codes with different hull dimensions}
\label{sec:proofs}

In this section, we evaluate the ratio $A_{n,k,\ell,q}/A_{n,k,\ell+1,q}$ for fixed values of $q$, $n$ and $k$, such that $\ell+1\le k\le n/2$.
According to \cite{LiShiLing}, we consider four cases for odd $q$ and four cases for even $q$.

\subsection{Odd $q$}

We consider four cases according to \eqref{oddq}.

\begin{itemize}
\item Let $n$ be even, $k-\ell$ be odd. In this case we must have in mind that if $q\equiv 3\pmod 4$, $n\equiv 2\pmod 4$ and $\ell=k=n/2$, then $A_{n,n/2,n/2,q}=\sigma_{n,n/2}=0$.
Therefore, for such values of $n$ and $q$, we take $\ell+1<k$ if $k=n/2$.
Now $k_0'=k-\ell-1=k_0-1$ is even and
\begin{align*}
  \frac{A_{n,k,\ell,q}}{A_{n,k,\ell+1,q}} & =\left(\prod_{i=1}^\ell\frac{q^{n-k_0-2i+1}-1}{q^{k_0}(q^i-1)}\right)q^{(k_0(n-k_0)-1)/2}B_1{n/2-1\brack(k_0-1)/2 }_{q^2}\\ &/\left(\prod_{i=1}^{\ell+1}\frac{q^{n-k_0+1-2i+2}-1}{q^{k_0-1}(q^i-1)}\right)q^{(k_0-1)(n-k_0+1)/2}B_2'{n/2\brack (k_0-1)/2 }_{q^2} \\
   &=\left(\prod_{i=1}^\ell\frac{(q^{n-k_0-2i+1}-1)}{(q^{n-k_0-2i+3}-1)q}\right)\frac{q^{k_0-1}(q^{\ell+1}-1)q^{(k_0(n-k_0)-1)/2}B_1}{(q^{n-k_0-2\ell+1}-1)q^{(k_0-1)(n-k_0+1)/2}B_2'}\\
   &~~~{n/2-1\brack(k_0-1)/2 }_{q^2}/{n/2\brack (k_0-1)/2 }_{q^2} \\
   &=\left(\prod_{i=1}^\ell\frac{q^{n-k_0-2i+1}-1}{q^{n-k_0-2i+3}-1}\right)\frac{q^{n/2-1}(q^{\ell+1}-1)B_1}{q^{\ell}(q^{n-k-\ell+1}-1)B_2'}\times\frac{{n/2-1\brack(k_0-1)/2 }_{q^2}}{{n/2\brack (k_0-1)/2 }_{q^2}}.
\end{align*}

Consider separately the multipliers in the last formula.
$$\prod_{i=1}^\ell\frac{q^{n-k_0-2i+1}-1}{q^{n-k_0-2i+3}-1}=\frac{q^{n-k_0-2\ell+1}-1}{q^{n-k_0+1}-1}=\frac{q^{n-k-\ell+1}-1}{q^{n-k_0+1}-1}$$
$$\frac{{n/2-1\brack(k_0-1)/2 }_{q^2}}{{n/2\brack (k_0-1)/2 }_{q^2}}=\frac{q^{n-k_0+1}-1}{q^n-1}$$
$$\frac{B_2'}{B_1}=\frac{q^{n/2-\ell-1}+\eta((-1)^{n/2})}{(q^{n/2}+\eta((-1)^{n/2}))(q^{n/2}-\eta((-1)^{n/2}))}=\frac{q^{n/2-\ell-1}+\eta((-1)^{n/2})}{q^{n}-1}$$

It follows that
$$\frac{A_{n,k,\ell,q}}{A_{n,k,\ell+1,q}}=\frac{q^{n/2-1}(q^{\ell+1}-1)}{q^{\ell}(q^{n/2-\ell-1}+\eta((-1)^{n/2}))}=\frac{q^{n/2-1}(q^{\ell+1}-1)}{q^{n/2-1}+\eta((-1)^{n/2})q^\ell}.$$

We observe that the considered ratio does not depend on $k$ but only on $l$, $q$, and the length $n$.

If $n\equiv 0\pmod 4$, or $n\equiv 2\pmod 4$ and $q\equiv 1\pmod 4$, then $\eta((-1)^{n/2})=1$. Hence
$$\frac{A_{n,k,\ell,q}}{A_{n,k,\ell+1,q}}= \frac{q^{n/2-1}}{q^{n/2-1}+q^\ell}(q^{\ell+1}-1)=(1-\frac{q^{\ell}}{q^{n/2-1}+q^\ell})(q^{\ell+1}-1)\ge \frac{q^{\ell+1}-1}{2},$$
as the equality holds in the case when $\ell=n/2-1$.

Otherwise, if $n\equiv 2\pmod 4$ and $q\equiv 3\pmod 4$, then $\eta((-1)^{n/2})=-1$, therefore
$$\frac{A_{n,k,\ell,q}}{A_{n,k,\ell+1,q}}= \frac{q^{n/2-\ell-1}}{q^{n/2-\ell-1}-1}(q^{\ell+1}-1)> q^{\ell+1}-1.$$

\item Let $n$ and $k_0$ be odd. Then $k_0'=k_0-1$ is even and so 
\begin{align*}
  \frac{A_{n,k,\ell,q}}{A_{n,k,\ell+1,q}} & =\left(\prod_{i=1}^{\ell}\frac{q^{n-k_0-2i+2}-1}{q^{k_0}(q^i-1)}\right)q^{(n-k_0)(k_0+1)/2-\ell}{(n-1)/2\brack(k_0-1)/2 }_{q^2}\\ &/\left(\prod_{i=1}^{\ell+1}\frac{q^{n-k_0-2i+2}-1}{q^{k_0-1}(q^i-1)}\right)q^{(k_0-1)(n-k_0+2)/2}{(n-1)/2\brack (k_0-1)/2 }_{q^2} \\
  & =\frac{q^{n-k-\ell}(q^{\ell+1}-1)}{q^{n-k-\ell}-1}=q^{\ell+1}-1+\frac{q^{\ell+1}-1}{q^{n-k-\ell}-1}>q^{\ell+1}-1.
\end{align*}

\item Let $n$ be odd and $k-\ell$ be even. Then $k_0'=k_0-1$ is odd and  
\begin{align*}
  \frac{A_{n,k,\ell,q}}{A_{n,k,\ell+1,q}} & =\left(\prod_{i=1}^\ell\frac{q^{n-k_0-2i+1}-1}{q^{k_0}(q^i-1)}\right)q^{k_0(n-k_0+1)/2}{(n-1)/2\brack k_0/2 }_{q^2}\\ &/\left(\prod_{i=1}^{\ell+1}\frac{q^{n-k_0-2i+3}-1}{q^{k_0-1}(q^i-1)}\right)q^{(n-k_0+1)k_0/2-\ell-1}{(n-1)/2\brack(k_0-2)/2 }_{q^2} \\
  &=\left(\prod_{i=1}^\ell\frac{q^{n-k_0-2i+1}-1}{q(q^{n-k_0-2i+3}-1)}\right)\frac{q^{k_0-1}(q^{\ell+1}-1)}{(q^{n-k_0-2\ell+1}-1)}q^{\ell+1}\frac{{(n-1)/2\brack k_0/2 }_{q^2}}{{(n-1)/2\brack k_0/2-1 }_{q^2}}\\
  &=\frac{q^{n-k_0-2\ell+1}-1}{(q^{n-k_0+1}-1)}\frac{q^{k}(q^{\ell+1}-1)}{q^\ell(q^{n-k_0-2\ell+1}-1)}\frac{q^{n-k_0+1}-1}{q^{k_0}-1}\\
  &=\frac{q^{k-\ell}(q^{\ell+1}-1)}{q^{k-\ell}-1}=q^{\ell+1}-1+\frac{q^{\ell+1}-1}{q^{k-\ell}-1}.
\end{align*}

Since $\frac{q^{\ell+1}-1}{q^{k-\ell}-1}>0$, the considered ratio is greater than $q^{\ell+1}-1$. But if $2\ell+1\ge k$, then $\frac{A_{n,k,\ell,q}}{A_{n,k,\ell+1,q}}\ge q^{\ell+1}$. If $k=2\ell+1$, then $A_{n,2\ell+1,\ell,q}/A_{n,2\ell+1,\ell+1,q}=q^{\ell+1}.$

\item Let $n$ and $k-\ell$ be even. Then $k_0'=k_0-1$ is odd and so 
\begin{align*}
  \frac{A_{n,k,\ell,q}}{A_{n,k,\ell+1,q}} & =\left(\prod_{i=1}^{\ell}\frac{q^{n-k_0-2i+2}-1}{q^{k_0}(q^i-1)}\right)q^{k_0(n-k_0)/2}B_2{n/2\brack k_0/2 }_{q^2}\\ &/\left(\prod_{i=1}^{\ell+1}\frac{q^{n-k_0-2i+2}-1}{q^{k_0-1}(q^i-1)}\right)q^{((k_0-1)(n-k_0+1)-1)/2}B_1{n/2-1\brack(k_0-2)/2 }_{q^2} \\
  &=\frac{q^{k_0-1}(q^{\ell+1}-1)q^{n/2-k_0+1}B_2{n/2\brack k_0/2 }_{q^2}}{q^\ell(q^{n-k_0-2\ell}-1)B_1{n/2-1\brack k_0/2-1 }_{q^2}}
\end{align*}

Using that $${n/2\brack k_0/2 }_{q^2}=\frac{q^n-1}{q^{k_0}-1}{n/2-1\brack k_0/2-1 }_{q^2},$$
we obtain
$$\frac{A_{n,k,\ell,q}}{A_{n,k,\ell+1,q}}=\frac{q^{n/2-\ell}(q^{n/2-\ell}+\eta((-1)^{n/2}))}{(q^{n-k-\ell}-1)(q^{k-\ell}-1)}(q^{\ell+1}-1)>q^{\ell+1}-1.$$
\end{itemize}

We summarize the results in the following lemma.

\begin{lemma}\label{lem-oddq} Let $q$ be a power of an odd prime, $n$ and $k$ be positive integers with $k\le n/2$,and $\ell$ be an integer such that $0\le\ell\le k-1$. Then
$$A_{n,k,\ell,q}=\alpha_{n,k,\ell,q}(q^{\ell+1}-1)A_{n,k,\ell+1,q},$$
where
\[
\alpha_{n,k,\ell,q}=\left\{\begin{array}{ll}
\frac{q^{n/2-1}}{q^{n/2-1}+\eta((-1)^{n/2})q^\ell}&\mbox{if} \ n \ \mbox{is even}, k-\ell \ \mbox{is odd},\\
\frac{q^{n-k-\ell}}{q^{n-k-\ell}-1}&\mbox{if} \ n \ \mbox{is odd}, k-\ell \ \mbox{is odd},\\
\frac{q^{k-\ell}}{q^{k-\ell}-1}&\mbox{if} \ n \ \mbox{is odd}, k-\ell \ \mbox{is even},\\
\frac{q^{n/2-\ell}(q^{n/2-\ell}+\eta((-1)^{n/2}))}{(q^{n-k-\ell}-1)(q^{k-\ell}-1)}&\mbox{if} \ n \ \mbox{is even}, k-\ell \ \mbox{is even}.
\end{array}\right.
\]
If $k-\ell$ is odd, $n\equiv 0\pmod 4$ or $n\equiv 2\pmod 4$ and $q\equiv 1\pmod 4$, $\alpha_{n,k,\ell,q}\ge\frac{1}{2}$ with equality when $k=n/2$ and $\ell=k-1$. In all other cases $\alpha_{n,k,\ell,q}>1$.
\end{lemma}

\begin{example} We give two examples with ternary codes:
\begin{enumerate}
\item Let $q=3$, $n=8$ and $k=4$. Then $A_{8,4,0,3}=48,958,182$,
$A_{8,4,1,3}=23,587,200$, $A_{8,4,2,3}= 3,276,000$, $A_{8,4,3,3}=89,600$, and $A_{8,4,4,3}=2240$. This gives us that
$$\frac{A_{8,4,0,3}}{A_{8,4,1,3}}=2.07563, \ \frac{A_{8,4,1,3}}{A_{8,4,2,3}}=7.2, \  \frac{A_{8,4,2,3}}{A_{8,4,3,3}}=36.5625, \ \frac{A_{8,4,3,3}}{A_{8,4,4,3}}=40=\frac{3^4-1}{2}.$$

\item Let $q=3$, $n=9$ and $k=4$. Then $A_{9,4,0,3}=3,965,612,742$,
$A_{9,4,1,3}=1,958,327,280$, $A_{9,4,2,3}=241,768,800$, $A_{9,4,3,3}=8,265,600$, and $A_{9,4,4,3}=91840$. This gives us that
$$\frac{A_{9,4,0,3}}{A_{9,4,1,3}}=2.025, \ \frac{A_{9,4,1,3}}{A_{9,4,2,3}}=8.1, \  \frac{A_{9,4,2,3}}{A_{9,4,3,3}}=29.25, \ \frac{A_{9,4,3,3}}{A_{9,4,4,3}}=90.$$
\end{enumerate}
\end{example}

In the next proposition, we give the values of $\ell$ for which $\mu_{n,k,\ell,q}=q^{\ell+1}-1$ for any of the options for the length $n$, when $q$ is odd.

\begin{proposition}\label{prop:oddq}
Let $q$ and $n$ be odd. Then $\mu_{n,k,\ell,q}=q^{\ell+1}-1$, if (1) $k-\ell$ is even and $\ell<\frac{k-1}{2}$, or (2) $k-\ell$ is odd and $\ell<\frac{n-k-1}{2}$.
Moreover, if $n\equiv 2\pmod 4$ and $q\equiv 3\pmod 4$, then $\mu_{n,k,\ell,q}=q^{\ell+1}-1$ if $\ell<\frac{n}{4}-1$.
\end{proposition}

\begin{proof} In these cases $\alpha_{n,k,\ell,q}=\frac{q^m}{q^m-1}>1$, where
\[
m=\left\{\begin{array}{ll}
k-\ell&\mbox{if} \ n \ \mbox{is odd}, k-\ell \ \mbox{is even},\\
n-k-\ell&\mbox{if} \ n \ \mbox{is odd}, k-\ell \ \mbox{is odd},\\
n/2-\ell-1&\mbox{if} \ n\equiv 2\pmod 4, k-\ell \ \mbox{is odd}, \mbox{and} \ q\equiv 3\pmod 4.
\end{array}\right.
\]

Recall that $\mu_{n,k,\ell,q}$ is the largest integer such that
$$A_{n,k,\ell,q}=\alpha_{n,k,\ell,q}(q^{\ell+1}-1)A_{n,k,\ell+1,q} \ge \mu_{n,k,\ell,q} A_{n,k,\ell+1,q}.$$
Hence $\mu_{n,k,\ell,q}$ is the largest integer such that
$\alpha_{n,k,\ell,q}(q^{\ell+1}-1) \ge \mu_{n,k,\ell,q} $.
It follows that $\mu_{n,k,\ell,q}=q^{\ell+1}-1$ when $\alpha_{n,k,\ell,q}(q^{\ell+1}-1)<q^{\ell+1}$. So
$$\frac{q^m}{q^m-1}(q^{\ell+1}-1)<q^{\ell+1}.$$
This inequality holds when $\ell+1<m$. If $m=k-\ell$ then $\ell+1<m$ in the case when $\ell<\frac{k-1}{2}$. In the case when $m=n-k-\ell$ the inequality $\ell+1<m$ holds in the case $\ell<\frac{n-k-1}{2}$. Finally, if $m=n/2-\ell-1$, then $\ell+1<m$ when $\ell<\frac{n}{4}-1$.
\end{proof}

\subsection{Even $q$}

Similar to the previous subsection, we consider four cases as  in \eqref{evenq}.

\begin{itemize}
\item Let $n$ be even, $k-\ell$ be odd. Now, if $n=2k$ and $k-\ell=1$ then the factor $q^{n-k-\ell-1}-1$ in the denominator is equal to 0 and therefore we cannot use the formula from Corollary \ref{cor-evenq}. Therefore, we consider this case separately.

    If $n>2k$, or $n=2k$ and $k-\ell>1$, then
\begin{align*}
  A_{n,k,\ell,q}/A_{n,k,\ell+1,q} & =\left(\prod_{i=1}^\ell\frac{q^{n-k_0-2i+1}-1}{(q^i-1)q^{k_0}}\right)q^{(nk_0-k_0^2+n-1)/2}{n/2-1\brack(k_0-1)/2 }_{q^2}\\ &/\left(\prod_{i=1}^{\ell+1}\frac{q^{n-k_0+1-2i}-1}{(q^i-1)q^{k_0-2}}\right)q^{(k_0-1)(n-k_0+1)/2}\frac{q^{n-\ell-1}-1}{q^{\ell+1}(q^{n-k-\ell-1}-1)}{n/2-1\brack (k_0-1)/2 }_{q^2} \\
  & =\frac{(q^{\ell+1}-1)q^{k_0-2}q^{\ell+1}(q^{n-k-\ell-1}-1)q^{(nk_0-k_0^2+n-1)/2}}{q^{2\ell}(q^{n-k-\ell-1}
  -1)(q^{n-\ell-1}-1)q^{(k_0-1)(n-k_0+1)/2}}\\
  & =\frac{(q^{\ell+1}-1)q^{k_0-2}q^{n-k_0}}{q^{\ell-1}(q^{n-\ell-1}-1)}= \frac{(q^{\ell+1}-1)q^{n-\ell-1}}{q^{n-\ell-1}-1}>q^{\ell+1}-1.
\end{align*}

If $n=2k$ and $k-\ell=1$, we have
\begin{align*}
  A_{2k,k,k-1,q}/A_{2k,k,k,q} & =\left(\prod_{i=1}^{k-1}\frac{q^{2k-2i}-1}{(q^i-1)q}\right)q^{n-1} /\left(\prod_{i=1}^{k}\frac{q^{2k-2i+2}-1}{q^i-1}\right)\frac{1}{q^k+1} \\
  & =\frac{q^{2k-1}(q^k+1)(q^k-1)}{q^{k-1}(q^{2k}-1)}=q^k=q^{\ell+1}.
\end{align*}

\item Let $n$ and $k_0$ be odd. In this case $n>2k> k+\ell$ and therefore we can freely use the formula from the corollary. Now $k_0'=k_0-1$ is even and so 
\begin{align*}
  A_{n,k,\ell,q}/A_{n,k,\ell+1,q} & =\left(\prod_{i=1}^{\ell}\frac{q^{n-k_0-2i}-1}{(q^i-1)q^{k_0-1}}\right)q^{(n-k_0)(k_0-1)/2+n-k}\frac{q^{n-k+\ell}-1}{q^{\ell}(q^{n-k-\ell}-1)}{(n-1)/2\brack(k_0-1)/2 }_{q^2}\\ &/\left(\prod_{i=1}^{\ell+1}\frac{q^{n-k_0-2i+2}-1}{(q^i-1)q^{k_0-1}}\right)q^{(k_0-1)(n-k_0+2)/2}{(n-1)/2\brack (k_0-1)/2 }_{q^2} \\
  & = \left(\prod_{i=1}^{\ell}\frac{q^{n-k_0-2i}-1}{q^{n-k_0-2i+2}-1}\right)\frac{(q^{\ell+1}-1)q^{k_0-1}}{q^{n-k_0-2\ell}-1}q^{\ell+1+n-2k} \frac{q^{n-k+\ell}-1}{q^{\ell}(q^{n-k-\ell}-1)}\\
   & = \frac{(q^{\ell+1}-1)q^{n-k}}{q^{\ell}(q^{n-k-\ell}-1)} =\frac{q^{n-k-\ell}}{q^{n-k-\ell}-1}(q^{\ell+1}-1)>q^{\ell+1}-1
\end{align*}

\item Let $n$ be odd and $k-\ell$ be even. Then $k_0'=k_0-1$ is odd, but since $k+\ell+1\le 2k<n$, we can apply the formulae from Corollary \ref{cor-evenq}.

\begin{align*}
  A_{n,k,\ell,q}/A_{n,k,\ell+1,q} & =\left(\prod_{i=1}^\ell\frac{q^{n-k_0-2i+1}-1}{(q^i-1)q^{k_0}}\right)q^{k_0(n-k_0+1)/2}{(n-1)/2\brack k_0/2 }_{q^2}\\ &/\left(\prod_{i=1}^{\ell+1}\frac{q^{n-k_0+1-2i}-1}{(q^i-1)q^{k_0-2}}\right)q^{(n-k_0+1)(k_0-2)/2+n-k}\frac{q^{n-k+\ell+1}-1}{q^{\ell+1}(q^{n-k-\ell-1}-1)}{(n-1)/2\brack(k_0-2)/2 }_{q^2} \\
  & =\frac{q^{\ell+1}(q^{n-k-\ell-1}-1)}{q^{2\ell}(q^{n-k+\ell+1}-1)}q^{\ell+1}\frac{(q^{\ell+1}-1)q^{k_0-2}}{q^{n-k-\ell-1}-1}\frac{q^{n-k+\ell+1}-1}{q^{k-\ell}-1} \\
  & =q^{k-\ell}\frac{(q^{\ell+1}-1)}{q^{k-\ell}-1}>q^{\ell+1}-1.
\end{align*}


\item Let $n$ and $k-\ell$ be even. Then $n\ge 2k>k+\ell$ and so $n-k-\ell>0$. Now $k_0'=k_0-1$ is odd and so 
\begin{align*}
  A_{n,k,\ell,q}/A_{n,k,\ell+1,q} & =\left(\prod_{i=1}^{\ell}\frac{q^{n-k_0-2i}-1}{(q^i-1)q^{k_0-1}}\right)q^{k_0(n-k_0)/2}\frac{q^{n-\ell}-1}{q^{\ell}(q^{n-k-\ell}-1)}{n/2-1\brack k_0/2 }_{q^2}\\ &/\left(\prod_{i=1}^{\ell+1}\frac{q^{n-k_0-2i+2}-1}{(q^i-1)q^{k_0-1}}\right)q^{k_0(n-k_0)/2+k_0-1}{n/2-1\brack(k_0-2)/2 }_{q^2} \\
  &=\left(\prod_{i=1}^{\ell}\frac{q^{n-k_0-2i}-1}{q^{n-k_0-2i+2}-1}\right)\frac{(q^{\ell+1}-1)(q^{n-\ell}-1)(q^{n-k+\ell}-1)}{(q^{n-k-\ell}-1)
  q^{\ell}(q^{n-k-\ell}-1)(q^{k-\ell}-1)}\\
   &=\frac{(q^{\ell+1}-1)(q^{n-\ell}-1)}{q^{\ell}(q^{n-k-\ell}-1)(q^{k-\ell}-1)}>q^{\ell+1}-1
\end{align*}
\end{itemize}

In this way, we proved the following lemma.

\begin{lemma}\label{lem-evenq} Let $q$ be a power of $2$, $n$ and $k$ be positive integers with $k\le n/2$,and $\ell$ be an integer such that $0\le\ell\le k-1$. Then
$$A_{n,k,\ell,q}=\alpha_{n,k,\ell,q}(q^{\ell+1}-1)A_{n,k,\ell+1,q},$$
where
\[
\alpha_{n,k,\ell,q}=\left\{\begin{array}{ll}
\frac{q^{n-\ell-1}}{q^{n-\ell-1}-1}&\mbox{if} \ n \ \mbox{is even}, k-\ell \ \mbox{is odd},\\
\frac{q^{n-k-\ell}}{q^{n-k-\ell}-1}&\mbox{if} \ n \ \mbox{is odd}, k-\ell \ \mbox{is odd},\\
\frac{q^{k-\ell}}{q^{k-\ell}-1}&\mbox{if} \ n \ \mbox{is odd}, k-\ell \ \mbox{is even},\\
\frac{q^{n-\ell}-1}{q^{\ell}(q^{n-k-\ell}-1)(q^{k-\ell}-1)}&\mbox{if} \ n \ \mbox{is even}, k-\ell \ \mbox{is even}.
\end{array}\right.
\]
In all cases $\alpha_{n,k,\ell,q}>1$.
\end{lemma}

\begin{example}\label{ex-evenq} We present three examples:
\begin{enumerate}
\item Let $q=2$, $n=10$ and $k=5$. Then $A_{10,5,0,2}=46,792,704$,
$A_{10,5,1,2}=46,701,312$, $A_{10,5,2,2}= 13,708,800$, $A_{10,5,3,2}=1,943,100$, $A_{10,5,4,2}=73440$, and $A_{10,5,5,2}=2295$. This gives us that
$$\frac{A_{10,5,0,2}}{A_{10,5,1,2}}=1.00196, \ \frac{A_{10,5,1,2}}{A_{10,5,2,2}}=3.40667, \  \frac{A_{10,5,2,2}}{A_{10,5,3,2}}=7.05512, \ \frac{A_{10,5,3,2}}{A_{10,5,4,2}}=26.4583,$$
$$\frac{A_{10,5,4,2}}{A_{10,5,5,2}}=32=2^5.$$

\item Let $q=2$, $n=9$ and $k=4$. Then $A_{9,4,0,2}=1,462,272$,
$A_{9,4,1,2}=1,370,880$, $A_{9,4,2,2}=428,400$, $A_{9,4,3,2}=45900$, and $A_{9,4,4,2}=2295$. This gives us that
$$\frac{A_{9,4,0,2}}{A_{9,4,1,2}}=1.06667, \ \frac{A_{9,4,1,2}}{A_{9,4,2,2}}=3.2, \  \frac{A_{9,4,2,2}}{A_{9,4,3,2}}=9.33333, \ \frac{A_{9,4,3,2}}{A_{9,4,4,2}}=20.$$

\item Let $q=4$, $n=8$ and $k=4$. Then $A_{8,4,0,4}=4,598,071,296$,
$A_{8,4,1,4}=1,520,762,880$, $A_{8,4,2,4}=101,359,440$, $A_{8,4,3,4}=1,414,400$, and $A_{8,4,4,4}=5525$. This gives us that
$$\frac{A_{8,4,0,4}}{A_{8,4,1,4}}=3.02353, \ \frac{A_{8,4,1,4}}{A_{8,4,2,4}}=15.0037, \  \frac{A_{8,4,2,4}}{A_{8,4,3,4}}=71.6625, \ \frac{A_{8,4,3,4}}{A_{8,4,4,4}}=256=4^4.$$

\end{enumerate}
\end{example}

In the following proposition, we present the relationship between the coefficients $\mu_{n,k,\ell,q}$ and $\alpha_{n,k,\ell,q}$ when $q$ is even. Its proof is similar to the proof of Proposition \ref{prop:oddq}.

\begin{proposition}
If $n$ is even and $k-\ell$ is odd, then $\mu_{n,k,\ell,q}=q^{\ell+1}-1$ unless $n=2k$ and $\ell=k-1$, when $\mu_{n,k,\ell,q}=\alpha_{n,k,\ell,q}=q^{\ell+1}$.

If $n$ is odd and $k-\ell$ is odd, then $\mu_{n,k,\ell,q}=q^{\ell+1}-1$ when $\ell<\frac{n-k-1}{2}$, otherwise $\mu_{n,k,\ell,q}\ge q^{\ell+1}$.

If $n$ is odd and $k-\ell$ is even, then $\mu_{n,k,\ell,q}=q^{\ell+1}-1$ when $\ell<\frac{k-1}{2}$, otherwise $\mu_{n,k,\ell,q}\ge q^{\ell+1}$.
\end{proposition}

In Example \ref{ex-evenq}, if $q=2$, $n=9$, $k=4$, $\ell=3$, we have $\mu_{9,4,3,2}= 20$, and if $q=2$, $n=9$, $k=4$, $\ell=2$, we have $\mu_{9,4,2,2}= 9$.

We summarise the results of this section in the following theorem.

\begin{theorem}\label{thm:ineq}
Let $\ell+1\le k\le n-\ell-1$. Then
$A_{n,k,\ell,q}>(q^{\ell+1}-1)A_{n,k,\ell+1,q}$
for all possible values of $q$, $n$, $k$ and $\ell$ except  when $q$ is odd, $n\equiv 2\pmod 4$, $-1$ is not a square in $\F_q^*$, and $k-\ell$ is odd, in which case
$A_{n,k,\ell,q}\ge\frac{q^{\ell+1}-1}{2}A_{n,k,\ell+1,q}$.
\end{theorem}

\section{Computational results}
\label{Results}

In this section, we present computational results with the number of inequivalent linear $[n,k,\ge 2]$  codes of different types for given length and dimension over $\F_2$ and $\F_3$, whose dual distance is at least 2 (this means that their generator matrices have no zero columns). Denote by $\mathcal{B}_q(n,k,\ell)$ the set of all inequivalent $q$-ary linear codes whose hull has dimension $\ell$. In the binary case, the mass formula, that helps to verify classification results for these codes, is the following
\begin{equation}
A_{n,k,\ell,2}=\sum_{C\in \mathcal{B}_2(n,k,\ell)}\frac{n!}{|\mathrm{Aut}(\mathcal{C})|}
\end{equation}

In the ternary case we have
\begin{equation}
A_{n,k,\ell,3}=\sum_{C\in \mathcal{B}_3(n,k,\ell)}\frac{2^n n!}{|\mathrm{Aut}(\mathcal{C})|}.
\end{equation}

More information on these formulae for self-dual codes can be found in \cite{HP}. Araya and Harada in \cite{Araya-Harada2019} used them in the classification of the binary LCD codes of length $n\le 13$ and ternary LCD codes of length $n\le 10$. They described clearly how to use these mass formulae for the verification of the computational results for classification of the binary LCD $[6,3]$ codes with all possible minimum distances $d\ge 1$ and dual distances $d^\perp\ge 1$.

The classification of linear codes with the same lengths, namely $n\le 13$ for $q=2$ and $n\le 9$ for $q=3$, is given in \cite{LiShi} and \cite{LiShiLing}, respectively. The difference compared to \cite{Araya-Harada2019} is that the authors also present the number of codes with different hull dimensions. Examining the tables in these papers, we notice that even then the number of inequivalent $[n,k]$ codes decreases as the dimension of the hull increases for $n\ge 2k$. The only exceptions are with the codes with hull dimension 0 and 1 in the binary case, for example $|\mathcal{B}_2(11,4,0)|=348<|\mathcal{B}_2(11,4,1)|=420$. The tables in \cite{LiShi} and \cite{LiShiLing} and our classification results, listed in Table \ref{table1}, give us reason to present the following hypothesis.

\begin{conjecture}
Let $B_{n,k,\ell,q}$ be the number of all inequivalent linear codes of length $n$, dimension $k$ and hull dimension $\ell$ over $\F_q$. If $q=2$ or 3 and $n\ge 2k$, then
$$\min\{B_{n,k,0,q},B_{n,k,1,q}\}>B_{n,k,2,q}>\cdots>B_{n,k,k,q}.$$
\end{conjecture}

We do not include $q>3$ in this conjecture because of the result in \cite{Carlet_Pellikaan} that any linear code over $\F_q$ for $q > 3$ is equivalent to a Euclidean LCD code.

\begin{table}
\caption{Classification of linear codes with different hull dimensions\label{table1}}
{\footnotesize
\begin{tabular}{r|c|c|c|c|c|c|c|c|c|c|c|c|c|c|c}
\hline\noalign{\smallskip}
\multicolumn{16}{c}{$q=2$, $k=4$}\\
\hline
\multicolumn{6}{c|}{$n=13$}& \multicolumn{5}{c|}{$n=15$}& \multicolumn{5}{c}{$n=17$}\\
\hline
$\ell=$    &  0 &  1& 2& 3& 4&0& 1& 2& 3& 4&0& 1& 2& 3& 4\\
\hline
 &	1363&	1635&	830&	200&	36&	4876&5704&2761&597&98&	16092&	18222&	8363& 1638&	245\\
\noalign{\smallskip}\hline
\multicolumn{6}{c|}{$n=14$}& \multicolumn{5}{c|}{$n=16$}& \multicolumn{5}{c}{$n=18$}\\
\hline
$\ell=$    &  0 &  1& 2& 3& 4&0& 1& 2& 3& 4&0& 1& 2& 3& 4\\
\hline
 &2733&	2835&	1710& 288&	75&	9265&	9664&	5284&	850&	190&29160&30171&15147&2323&444\\
\noalign{\smallskip}\hline
\end{tabular}
\begin{tabular}{r|c|c|c|c|c|c|c|c|c|c|c|c}
\hline\noalign{\smallskip}
\multicolumn{13}{c}{$q=2$, $k=5$}\\
\hline
\multicolumn{7}{c|}{$n=13$}& \multicolumn{6}{c}{$n=15$}\\
\hline
$\ell=$    &  0 &  1& 2& 3& 4&5&0& 1& 2& 3& 4&5\\
\hline
 &	4576&	5943&	3065&	950&	167&	23&33711&42016&19799&5413&	785&	94\\
\noalign{\smallskip}\hline
\multicolumn{7}{c|}{$n=14$}& \multicolumn{6}{c}{$n=16$}\\
\hline
$\ell=$    &  0 &  1& 2& 3& 4&5&0& 1& 2& 3& 4&5\\
\hline
 &12103&16798&7142&2606&	296&	61&	88102&	112633&	45837&13722&1402&228\\
\noalign{\smallskip}\hline
\end{tabular}
\begin{tabular}{r|c|c|c|c|c|c|c|c|c|c|c|c|c|c}
\hline\noalign{\smallskip}
\multicolumn{15}{c}{$q=2$, $k=6$}\\
\hline
\multicolumn{8}{c|}{$n=13$}& \multicolumn{7}{c}{$n=14$}\\
\hline
$\ell=$    &  0 &  1& 2& 3& 4&5&6&0& 1& 2& 3& 4&5&6\\
\hline
 &	9036&	11799&	6425&	2007&	432&60&6&37982&44759&26340&6265&	1661&135&27\\
\noalign{\smallskip}\hline
\end{tabular}
\begin{tabular}{r|c|c|c|c|c|c|c|c|c|c|c|c|c|c|c}
\hline\noalign{\smallskip}
\multicolumn{16}{c}{$q=3$, $k=4$}\\
\hline
\multicolumn{6}{c|}{$n=11$}& \multicolumn{5}{c|}{$n=12$}& \multicolumn{5}{c}{$n=13$}\\
\hline
$\ell=$    &  0 &  1& 2& 3& 4&0& 1& 2& 3& 4&0& 1& 2& 3& 4\\
\hline
 &	4511&	3096&	814&	115&	10&	16004&10337&2390&291&26&	57158&	34508&	7047& 723&	52\\
\noalign{\smallskip}\hline
\end{tabular}
\begin{tabular}{r|c|c|c|c|c|c|c|c|c|c|c|c}
\hline\noalign{\smallskip}
\multicolumn{13}{c}{$q=3$, $k=5$}\\
\hline
\multicolumn{7}{c|}{$n=11$}& \multicolumn{6}{c}{$n=12$}\\
\hline
$\ell=$    &  0 &  1& 2& 3& 4&5&0& 1& 2& 3& 4&5\\
\hline
 &	16769&	10942&	2567&	348&	37&	3&138865&82124&15881&1694&	137&	10\\
\noalign{\smallskip}\hline
\end{tabular}
\begin{tabular}{r|c|c|c|c|c|c|c}
\hline\noalign{\smallskip}
\multicolumn{8}{c}{$q=3$, $k=6$, $n=12$}\\
\hline
$\ell=$    &0& 1& 2& 3& 4&5&6\\
\hline
 &	314870&	179578&	32993&	3172&	267&	20&	3\\
\noalign{\smallskip}\hline
\end{tabular}
}
\end{table}

By the end of this section, we classify binary and ternary linear, self-orthogonal and LCD codes of a given dimension $3\le k\le 10$, length $k+3\le n\le 20$, minimum distance $d\ge 2$ and dual distance $d^\perp\ge 2$.
 In the binary case, we classify also the even codes with the corresponding length and dimension, i.e. all linear codes whose codewords have only even weights. We do not count codes with dual distance $d^\perp=1$ because if $\mathcal{C}$ is an $[n,k,d]$ code with $d^\perp=1$, then all its codewords share a common zero coordinate, so $\mathcal{C}=(0|\mathcal{C}_1)$ where $\mathcal{C}_1$ is an $[n-1,k,d]$ code. In this case, $\mathcal{C}^\perp=(0|\mathcal{C}_1^\perp)\cup (1|\mathcal{C}_1^\perp)$, $\hull(\mathcal{C})=(0|\hull(\mathcal{C}_1))$ and $\dim \hull(\mathcal{C})=\dim\hull(\mathcal{C}_1)$.
  If $\mathcal{C}$ is an $[n,k,1]$ code then $\mathcal{C}\cong (0|\mathcal{C}_1)\cup (1|\mathcal{C}_1)$, where $\mathcal{C}_1$ is an $[n-1,k-1]$ code, and then $\mathcal{C}^\perp\cong (0|\mathcal{C}_1^\perp)$. This gives us that $\hull(\mathcal{C})=(0|\hull(\mathcal{C}_1)$ and $\dim \hull(\mathcal{C})=\dim\hull(\mathcal{C}_1)$.
  Therefore, if we have the number $B_{n,k,0,q}^*$ of all $[n,k,\ge 2]$ LCD codes with dual distance $\ge 2$ for all length $\le n$ and dimensions $\le k$, we can easily compute the number $B_{n,k,0,q}$ of LCD $[n,k,\ge 1]$ codes with dual distance $\ge 1$, using the following formula
$$B_{n,k,0,q}=\sum_{m=k+1}^{n}B_{m,k,0,q}^*+B_{n,k-1,0,q}.$$

We give a simple example that can be followed by hand.

\begin{example}\rm Let $n=4$, $k=2$ and $q=2$. In this case
$$A_{4,2,2,2}=\sigma_{4,2}=3, \ A_{4,2,1,2}=12, \ A_{4,2,0,2}=20, \
B_{4,2,2,2}=B_{4,2,1,2}=1, \ B_{4,2,0,2}=4.$$

There are six inequivalent $[4,2]$ binary codes. The first one is obtained from $\F_2^2$ by adding two zero columns and it is an LCD code. There are two more binary $[4,2]$ inequivalent codes with zero columns, and these are the codes $\mathcal{C}_2=\{0000,0110,0001,0111\}$, $\dim \hull(\mathcal{C}_2)=1$, and $\mathcal{C}_3=\{0000,0110,0101,0011\}$, $\dim \hull(\mathcal{C}_3)=0$. The remaining three codes are $\mathcal{C}_4=\{0000,1110,0001,1111\}$, $\dim \hull(\mathcal{C}_4)=0$, $\mathcal{C}_5=\{0000,1110,0101,1011\}$, $\dim \hull(\mathcal{C}_5)=0$, and the self-dual $[4,2,2]$ code $\mathcal{C}_6=\{0000$, $1100,0011,1111\}$, $\dim \hull(\mathcal{C}_6)=2$. Only two of these codes have minimum and dual distance $d=d^\perp=2$ and these are $\mathcal{C}_5$ and $\mathcal{C}_6$.
\end{example}

We use the following properties of the considered types of binary and ternary codes:
\begin{itemize}
\item If $\mathcal{C}$ is an LCD code, its dual code $\mathcal{C}^\perp$ is also LCD. It follows that $B_{n,k,0,q}^*=B_{n,n-k,0,q}^*$ and $B_{n,k,0,q}=B_{n,n-k,0,q}$ both in the binary and the ternary case. The same holds for the number of all linear codes of length $n$ and dimension $k$.
\item Self-orthogonal codes exist only when $k\le n/2$.
\item All binary self-orthogonal codes are even.
\item Ternary self-dual codes exist only for lengths a multiple of 4 and only have codewords of Hamming weight a multiple of 3.
\end{itemize}

We obtain the classification results by the program \textsc{Generation} of the software package \textsc{QExtNewEdition} \cite{Generation}. The computations were executed on a Windows 11 OS in a single core of an Intel Xeon Gold 5118 CPU with a 2.30 GHz clock frequency.

In the binary case, we classify all linear, even, self-orthogonal and LCD codes of length $n\le 20$ and dimension $k\le 10$. The results are presented in Table \ref{table2}. For some values of $n$ and $k$, when all inequivalent codes are too many (more than a milion), we classify only optimal $[n,k]$ codes. In this cases, we put a $*$ after the number of codes. Consider for example $n=17$ and $k=7$. The largest possible minimum distance for a binary $[17,7]$ code is $d=6$ (see \cite{Grassl-table}). There are exactly 377 binary linear $[17,7,6]$ codes. Exactly 329 of these codes are even, 7 are LCD, but none of them is self-orthogonal. Furthermore, there are 497119 even, 58 self-orthogonal, and 14 734 654 LCD $[17,7,\ge 2]$ binary codes with dual distance $d^\perp\ge 2$.  Since we have the count of both all $[17,7,\ge 2]$ and optimal $[17,7,6]$ even codes without zero columns, we denote this in the table by 497119(329*). The optimal binary self-orthogonal codes with the parameters presented in Table \ref{table2} have been also classified in \cite{BBGO}.

The results in the table confirm the fact that LCD codes are much more than self-orthogonal and even more than even codes for a given length and size, even when considering only inequivalent codes. However, this is not the case if we consider only the optimal codes. In most cases, when the optimal code is unique, it is not LCD. The optimal $[19,7,8]$ is self-orthogonal, as is the optimal $[20,8,8]$ code. The optimal $[17,8,6]$ and $[17,9,5]$ codes are LCD, but the optimal $[19,8,7]$, $[20,9,7]$ and $[18,9,6]$ are neither self-orthogonal, nor LCD. We see interesting examples in the optimal codes of dimension 7. Out of all 377 optimal $[17,7,6]$ codes, none is self-orthogonal, but 7 are LCD codes. We have the opposite situation for length 20, namely out of all 26 optimal $[20,7,8]$ codes, none is LCD, but four are self-orthogonal. It is also worth noting the optimal $[20,10,6]$ codes, where out of all 1682 codes only one is odd-like (it contains codewords of odd weight), but none of the 1681 even codes is self-orthogonal.

In the ternary case, we consider linear, self-orthogonal and LCD codes. We classify all $[n,k,\ge 2]$ linear, LCD and SO codes without zero columns for $k=3$ and $n\le 20$, $k=4$ and $n\le 15$, $k=5$ and $n\le 14$, $6\le k\le 9$ and $n\le 13$. For $k=7$ and $n=15,16,17$, $k=8$ and $n=16$, 17, 18, and $k=9$, $n=19$ we classify only the self-orthogonal codes.  The question mark (?) in the table means that there are too many corresponding codes (more than a million).  Furthermore, we classify all optimal codes of these three types with 21 parameters. The results are presented in Table \ref{table3}. 

If we look at the number of codes with dimension 5 and length $n$, $10\le n\le 14$, we see that the LCD codes are more than half of all inequivalent $[n,5,\ge 2]$ linear codes. The situation with self-orthogonal codes is quite different - they occur much less often. For example, for length 13, the linear codes are more than a million, but only 17 of them are self-orthogonal. For the optimal codes, we have: (1) of four $[15,5,8]$ linear codes one is LCD (none is self orthogonal since 8 is not a multiple of 3), (2) the only $[16,5,9]$ code is self-orthogonal, (3) there are 1804 linear $[17,5,9]$ codes, 35 of which are self-orthogonal and 400 are LCD, (4) none of the seven $[18,5,10]$ and both $[19,5,11]$ codes is LCD, nor self-orthogonal, (5) there are two linear $[20,5,12]$ codes and both are self-orthogonal.

We have interesting results with the optimal $[19,7,9]_3$ and $[20,8,9]_3$ codes. In these cases, all optimal linear codes are self-orthogonal.

\begin{table}
\caption{Classification of binary linear codes $(d^\perp\ge 2)$\label{table2}}
{\footnotesize
\begin{tabular}{r|c|c|c|c|c|c|c|c|c|c|c|c|c|c|c}
\hline\noalign{\smallskip}
\multicolumn{16}{c}{$k=3$} \\
\hline
$n=$    & 6 & 7 & 8 & 9 & 10& 11 &  12& 13& 14& 15& 16& 17& 18& 19& 20 \\
\hline
linear  & 8 & 15& 27& 45& 71& 107& 159&	226&	317&	435&587&	779&1024&1325&1699\\
even    & 3 &  4&  8&  9& 17&  20&	34&	39&	61&	72&	106&	123&	174&	204&	277\\
SO      & 1 &  1&  3&  1&  6&   2&	12&	4&	21&	7&	34&	11&	54&	19&	82\\
LCD     & 2 &  5&  7& 17& 20&  42&	47&	91&	98&	180&	189&	328&	340&	565&	580\\
\noalign{\smallskip}\hline
\end{tabular}
\begin{tabular}{r|c|c|c|c|c|c|c|c|c|c|c|c|c|c}
\hline\noalign{\smallskip}
\multicolumn{15}{c}{$k=4$} \\
\hline
$n=$    & 7 & 8 & 9 & 10& 11 &  12& 13& 14& 15& 16& 17& 18& 19& 20 \\
\hline
linear  & 15& 42& 100&	222&	462&	928&	1782&	3333&	6058&	10759&	18694&	31877&	53357&	87864\\
even    &  4& 10&  18&	37&	63&	122&	202&	366&	602&	1038&	1671&	2785&	4411&	7122\\
SO      &  -&  2&   1&	6&	3&	16&	8&	39&	23&	92&	55&	199&	131&	424\\
LCD     &  5& 16&  30&	82&	139&	345&	568&	1267&	2040&	4193&	6631&	12720&	19734&	35732\\
\noalign{\smallskip}\hline
\end{tabular}
\begin{tabular}{r|c|c|c|c|c|c|c|c|c|c|c|c|c}
\hline\noalign{\smallskip}
\multicolumn{14}{c}{$k=5$} \\
\hline
$n=$    &  8 & 9 & 10& 11 &  12& 13& 14& 15& 16& 17& 18& 19& 20 \\
\hline
linear  & 	27&	100&	331&	1007&	2936&	8208&	22326&	59235&	153711&	390607&	972726&	2373644&	5676542\\
even    &  	7&	16&	46&	102&	264&	593&	1448&	3319&	7886&	18096&	42193&	96243&	219712\\
SO      &  	-&	-&	2&	2	&11&	8&	38&	33&	134&	123&	442&	462&	1450\\
LCD     &  	7&	30&	84&	297&	816&	2596&	6908&	20238&	52248&	142468&	355083&	908879&	2177772\\
\noalign{\smallskip}\hline
\end{tabular}
\begin{tabular}{r|c|c|c|c|c|c|c|c|c|c|c|c}
\hline\noalign{\smallskip}
\multicolumn{13}{c}{$k=6$} \\
\hline
$n=$    &  9 & 10& 11 &  12& 13& 14& 15& 16& 17& $[18,6,8]$& $[19,6,8]$& $[20,6,8]$\\
\hline
linear  & 	45&	222&	1007&	4393&	18621&	78148&	325815&	1350439&	5548052&2*	& 28* &1833*\\
even    &  	9&	30&	92&	303&	945&	3166&	10576&	37017&	131233&	2*&	21*&	1418*\\
SO      &  	-&	-&	-&	3&	3&	21&	21&	105&	123&	521(2*)&	746(2*)&	2758(23*)\\
LCD     &  	17&	82&	297&	1418&	5632&	25954&	108846&	484648&	2034711&	8633817(0*)	&	2* &392*\\
\noalign{\smallskip}\hline
\end{tabular}
\begin{tabular}{r|c|c|c|c|c|c|c|c|c|c|c}
\hline\noalign{\smallskip}
\multicolumn{12}{c}{$k=7$} \\
\hline
$n=$    &  10& 11 &  12& 13& 14& 15& 16& $[17,7,6]$& $[18,7,7]$& $[19,7,8]$)& $[20,7,8]$\\
\hline
linear  &		71&	462&	2936&	18621&	121169&	814087&	5635181&	377*&	2*&1*	&	26*\\
even    &		13&	46&	194&	774&	3518&	16714&	87998&	497119(329*)&	3010238(0*)&	1*&	21*\\
SO      &		-&	-&	-&	-&	4&	6&	41&	58(0*)&	300(0*)&	540(1*)&	2469(4*)\\
LCD     &		20&	139&	816&	5632&	37166&	272131&	1968462&	14734654(7*)&	0*&	0*&	0*\\
\noalign{\smallskip}\hline
\end{tabular}
\begin{tabular}{r|c|c|c|c|c|c|c|c|c|c}
\hline\noalign{\smallskip}
\multicolumn{11}{c}{$k=8$} \\
\hline
$n=$    &  11 &  12& 13& 14& 15& 16& $[17,8,6]$& $[18,8,6]$& $[19,8,7]$& $[20,8,8]$\\
\hline
linear  &		107&	928&	8208&	78148&	814087&	9273075&	1*&	918*&	1*&	1*\\
even    &		16&	76&	362&	2020&	12646&	94136&	818890(1*)&	907*&	0*&	1*\\
SO      &		-&	-&	-&	-&	-&	7&	10(0*)&	86(0*)&	168(0*)&	1016(1*)\\
LCD     &		42&	345&	2596&	25954&	272131&	3315862&  1*&	337*&	0*&	0*\\
\noalign{\smallskip}\hline
\end{tabular}
\begin{tabular}{r|c|c|c|c|c|c|c|c|c}
\hline\noalign{\smallskip}
\multicolumn{10}{c}{$k=9$} \\
\hline
$n=$    &   12& 13& 14& 15& 16& $[17,9,5]$& $[18,9,6]$& $[19,9,6]$& $[20,9,7]$\\
\hline
linear  & 	159&	1782&	22326&	325815&	5635181&	1*&	1*&	1700*&	1*\\
even    &  	22&	109&	689&	4973&	46344&	554238(0*)&	8547530(1*)&	1694*&	0*\\
SO      & 	-&	-&	-&	-&	-&	-&	9(0*)&	22(0*)&	194(0*)\\
LCD     & 	47&	568&	6908&	108846&	1968462&	1*&	0*&	3*&	0*\\
\noalign{\smallskip}\hline
\end{tabular}
\begin{tabular}{r|c|c|c|c|c|c|c|c}
\hline\noalign{\smallskip}
\multicolumn{9}{c}{$k=10$} \\
\hline
$n=$    &   13& 14& 15& 16& $[17,10,4]$& $[18,10,4]$& $[19,10,5]$& $[20,10,6]$\\
\hline
linear  & 	226	&3333&	59235&	1350439&	14390*&	11581361*&	31237*&	1682*\\
even    &  	26&	165&	1230&	12257&	169691(2614*)&	3433243(263147*)&	0*&	1681*\\
SO      &  	-	&-&	-&	-&	-&	-&	-&	16(0*)\\
LCD     &  	91&	1267&	20238&	484648&	14734654(4550*)&	4535834*&	11554*&	601*\\
\noalign{\smallskip}\hline
\end{tabular}
}
\end{table}

\begin{table}
\caption{Classification of ternary linear codes $(d^\perp\ge 2)$\label{table3}}
{\footnotesize
\begin{tabular}{r|c|c|c|c|c|c|c|c|c|c|c|c|c|c|c}
\hline\noalign{\smallskip}
\multicolumn{16}{c}{$k=3$} \\
\hline
$n=$    & 6 & 7 & 8 & 9 & 10& 11 &  12& 13& 14& 15& 16& 17& 18& 19& 20 \\
\hline
linear  &14&	31&	68&	137&	263&	484&	878&	1538&	2649&	4474&	7421&	12093&	19420&	30680&	47793\\
SO      & 0&	1&	1&	3&	4&	5&	10&	15&	17&	31&	44&	54&	91&	126&	160\\
LCD     & 7&	15&	33&	67&	132&	253&	471&	839&	1491&	2560&	4294&	7142&	11583&	18423&	29070\\
\noalign{\smallskip}\hline
\end{tabular}
\begin{tabular}{r|c|c|c|c|c|c|c|c|c|c|c|c}
\hline\noalign{\smallskip}
\multicolumn{13}{c}{$k=4$} \\
\hline
$n=$    & 7 & 8 & 9 & 10& 11 &  12& 13& 14& 15& $[16,4,9]$& $[17,4,10]$& $[18,4,11]$ \\
\hline
linear &31&	129&	460&	1638&	5701&	19996&	69536&	239681&	809694&	317*&	18*&	2*\\
SO &-&	1&	1&	3&	5&	16&	26&	52&	121&	255(13*)&	523&	1267\\
LCD &15&	33&	220&	839&	3077&	11228&	40668&	144447&	497679&	124*&	8*&	0*\\
\noalign{\smallskip}\hline
\end{tabular}
\begin{tabular}{r|c|c|c|c|c|c|c|c|c|c|c}
\hline\noalign{\smallskip}
\multicolumn{12}{c}{$k=5$} \\
\hline
$n=$    & 8 & 9 & 10& 11 &  12& 13& 14& $[15,5,8]$& $[16,5,9]$& $[17,5,9]$& $[18,5,10]$ \\
\hline
linear &68&	460&	3221&	24342&	202064&	1767647&	15604611&	4*&	1*&	1804*&	7*\\
SO &-&	-&	0&	3&	7&	17&	44&	156&	523(1*)&	1981(35*)&	9460\\
LCD &33&	220&	1681&	13537&	118878&	1080479&	9737965&	1*&	0*&	400*&	0*\\
\noalign{\smallskip}\hline
\end{tabular}
\begin{tabular}{r|c|c|c|c|c|c|c|c|c|c|c|c|c}
\hline\noalign{\smallskip}
\multicolumn{6}{c|}{$k=6$}&\multicolumn{8}{c}{$k=7$} \\
\hline
$n=$    & 9 & 10& 11 &  12& 13&10& 11 &  12& 13 & 15& 16& 17&20 \\
\hline
linear &	137&	1638&	24342&	474106&	10956955&263&	5701&	202064&	10956955&	?&	?&	?&?	\\
SO &-&	-&	-&	3&	4&	-&		-&	-&	-&	12&	50&	249&	2287775\\
LCD &	67&	839&	13537&	283650&	6807504&132&	3077&	118878&	6807504&	?&	?&	?&?	\\
\noalign{\smallskip}\hline
\end{tabular}
\begin{tabular}{r|c|c|c|c|c|c|c|c|c|c}
\hline\noalign{\smallskip}
\multicolumn{7}{c|}{$k=8$}& \multicolumn{4}{c}{$k=9$}\\
\hline
$n=$    &  11 &  12& 13& 16& 17& 18& 12& 13& 14& 19\\
\hline
linear &	484&	19996&	1767647&	?&	?&	?&	878	&69536&	15604611&	?\\
SO &-&		-&	-&	7&	16&	137&-&		-&	-&	56\\	
LCD &	253&	11228&	1080479&	?&	?&	?&	471&	40668&	11835111&	?\\
\noalign{\smallskip}\hline
\end{tabular}\\
\begin{tabular}{r|c|c|c|c|c|c|c|c}
\hline\noalign{\smallskip}
\multicolumn{9}{c}{Optimal codes} \\
\hline
    & $[19,4,12]$& $[20,4,12]$& $[19,5,11]$& $[20,5,12]$& $[14,6,6]$& $[15,6,7]$& $[16,6,7]$& $[17,6,8]$ \\
\hline
linear&	1*&	84*&	2*&	2*&	47674*&	22*&	$>10^8*$&2145181*\\
SO &2867(1*)&	6893(32*)&	50618&	294990(2*)&	15(4*)&	61&	286&	1504	\\
LCD &	0*&	14*&	0*&	0*&	27776*&	10*&	53236943*&	807993*	\\
\noalign{\smallskip}\hline
\end{tabular}
\begin{tabular}{r|c|c|c|c|c|c|c}
\hline\noalign{\smallskip}
\multicolumn{8}{c}{Optimal codes} \\
\hline
$n=$     & $[18,6,9]$& $[19,6,9]$ & $[18,7,8]$& $[19,7,9]$&  $[19,8,8]$& $[20,8,9]$ & $[20,9,8]$\\
\hline
linear &	171*&	?&	827459*&	61*&	1508*&	23*&	32*\\
SO &	13831(105*)&	184980(18019*)&	2486&	57551(61*)&2281&	112899(23*)&	1122\\
LCD &	4*&	?&	450403*&	0*&		363*&	0*&	2*\\
\noalign{\smallskip}\hline
\end{tabular}
}
\end{table}

\section{Conclusion}
\label{sec_conclusion}

By a result of Sendrier \cite{Sendrier_hull}, it is known that most linear codes are LCD when $q$ is large. Moreover, the proportion of $q$-ary linear codes of length $n$, dimension $k$ and specified hull dimension $\ell$ to all $q$-ary linear codes of the same length and dimension is convergent when $n$ and $k$ goes to infinity. Using the limit, Sendrier proved that the average dimension of the hull of a $q$-ary linear code is asymptotically equal to $\sum_{i\ge 1}\frac{1}{q^i+1}$ \cite{Sendrier_hull}.

In this paper, we obtain general results on hulls of linear codes, proving that the number of all $q$-ary linear codes of a given length $n$ and dimension $k$ decreases when the hull dimension increases, for all values of $n$, $k$ and $q$.

In addition, we classify all binary linear, even, self-orthogonal and LCD $[n\le 20, k\le 10, d\ge 2]$ codes and the ternary linear, self-orthogonal and LCD
$[n\le 19, k\le 10, d\ge 2]$ codes (with a few exceptions). For some considered values of $n$ and $k$, when the number of all inequivalent linear codes is huge, we classify only the optimal codes. The results are listed in Tables \ref{table2} and \ref{table3}.

\section*{Acknowledgments}
The research of Stefka Bouyuklieva is supported by Bulgarian National Science Fund grant number KP-06-H62/2/13.12.2022. The research of Iliya Bouyukliev is supported by project IC-TR/10/2024-2025. The research of Ferruh \"{O}zbudak is supported by T\"{U}B\.{I}TAK under Grant 223N065.
	
%
%


%
%



\end{document}